\documentclass[10pt, conference, letterpaper]{ IEEEtran}
\IEEEoverridecommandlockouts
\usepackage{cite}
\usepackage{amsmath,amssymb,amsfonts}
\usepackage{algorithm, algorithmic}
\usepackage{graphicx}
\usepackage{textcomp}
\usepackage{xcolor}
\usepackage{caption}
\usepackage{booktabs}
\usepackage{url}

\def\BibTeX{{\rm B\kern-.05em{\sc i\kern-.025em b}\kern-.08em
    T\kern-.1667em\lower.7ex\hbox{E}\kern-.125emX}}

\usepackage{amsthm}
\newtheoremstyle{bfnote}%
{}{}
{\itshape}{}
{\bfseries}{.}
{ }{\thmname{#1}\thmnumber{ #2}\thmnote{ (#3)}}
\theoremstyle{bfnote}

\newtheorem{theorem}{Theorem}

\newtheorem{lemma}{Lemma}
\newtheorem{remark}{Remark}
\newtheorem{definition}{Definition}
\newtheorem{proposition}{Proposition}

\newcommand{\QSent}{\mathbf{Q}} 
\newcommand{\Language}{\mathcal{L}} 
\newcommand{\ProbMeas}{\mathcal{P}_{I}} 

\begin{document}

\title{Goal-Oriented Logic-based Semantic Communication for Neuro-Symbolic Reasoning with Applications onto Autonomous Driving
}

\author{\IEEEauthorblockN{Ahmet Faruk Saz}
\IEEEauthorblockA{
\textit{Georgia Institute of Technology}\\
Atlanta, GA, USA \\
asaz3@gatech.edu}
\and
\IEEEauthorblockN{Duo Xu}
\IEEEauthorblockA{
\textit{Georgia Institute of Technology}\\
Atlanta, GA, USA \\
asaz3@gatech.edu}
\and
\IEEEauthorblockN{Faramarz Fekri}
\IEEEauthorblockA{
\textit{Georgia Institute of Technology}\\
Atlanta, GA, USA \\
faramarz.fekri@ece.gatech.edu}
}

\maketitle

\begin{abstract}
We consider First-Order Logic (FOL)-based semantic communication for neuro-symbolic decision-making in collaborative environments such as autonomous driving networks. Each connected autonomous vehicle (CAV) converts its partial sensor observations into a natural-language scene description and corresponding grounded FOL evidence. Under an uplink budget, a semantic encoder at each car selects the observations most informative for evaluating traffic rules and transmit to a Road Side Unit (RSU). The RSU fuses all received evidence, evaluates collaborative rules, performs logical deduction for vehicle-specific safety and right-of-way information for constrained downlink transmission. Each CAV combines the received deductions with its local description, enabling a local LLM agent to select a high-level driving action. We develop a principled, verifiable semantic communication method using a random-support Dirichlet--Categorical model of inductive logical probability, providing a modern statistical reinterpretation of Carnap's and Hintikka's systems. From this model, we derive a goal-oriented semantic information-bottleneck formulation that prioritizes evidence transmission by its reduction of uncertainty over task goals. Using 152 traffic rules extracted from the California Driver Handbook, we evaluate the framework on MDrive simulator in CARLA. Under identical communication budgets, semantic evidence selection completes every scenario without safety hazards, whereas uniform evidence selection produces collisions, showcasing semantic communication's superiority.
\end{abstract}


\section{Introduction}
\label{sec:introduction}

Semantic and goal-oriented communication shift the objective of a communication system from reconstructing raw data to conveying the information required for a downstream task. This shift is particularly important in collaborative autonomous driving, where occlusions, limited sensing ranges, and viewpoint-dependent visibility provide each connected autonomous vehicle (CAV) with only a partial view of the environment. Consider two vehicles approaching an unsignalized intersection where a truck obstructs their view of each other or of a crossing pedestrian. Such scenarios, represented in the NHTSA pre-crash typology~\cite{najm2007precrash}, pose genuine collision risks. By sharing complementary observations, the vehicles can be alerted to otherwise occluded road users and take appropriate action to avoid a collision. However, exchanging complete sensor streams can be prohibitively expensive under bandwidth and latency constraints. CAVs must instead communicate observations that materially affect safety and driving decisions.

Most semantic communication methods employ learned representations, including deep joint source--channel coding, functional compression, transformer-based encoders, LLMs, and generative models~\cite{Xie2020DeepLE,Yang2023SwinJSCCTS,Liu2024ANI,Wang2024LargeLanguageModelEnabledTS, Saidutta2022AML, yashas}. Although effective, their latent representations generally provide limited insight into how a transmitted feature affects the receiver's decision. Current theoretical approaches based on FOL, semantic rate--distortion, synonymous mappings, information bottlenecks, and rate--distortion--perception offer more principled objectives~\cite{Shao2022ATO,Niu2024AMT,Niu2025RateDistortionPerceptionTI,Saz2024OnTT,Saz2024LossySC,Saz2025DISCDDL,Saz2025AnalysisOS, saz2026goal}. Among these, however, only FOL-based approaches are specifically suited to the verifiable high-level reasoning and decision-making required in autonomous driving.

That is, human driving is not a single pattern-matching process as most deep learning based approaches assume: a driver first interprets the surrounding scene and then reasons over that interpretation using traffic rules, prior knowledge, and commonsense. Authors in \cite{Kothawade2021AUTODISCERNAD} argue that these functions require different computational mechanisms: neural models are well suited to perception and pattern recognition, whereas high-level decisions are more naturally expressed through symbolic reasoning. Current end-to-end driving systems often couple perception, functional compression, trajectory planning, and control within a single neural architecture, thereby obscuring the premises and reasoning that produce an action. This is especially problematic because driving may require deductive, inductive, abductive, default, counterfactual, causal, or commonsense reasoning, all of which draw conclusions from structured evidence and background knowledge. Other works reinforce this distinction: explicit commonsense layers improve the adaptability and explainability of learned driving systems~\cite{Kimbrell2025CommonsenseRA}; abductive reasoning enables vehicles to maintain hypotheses about objects hidden by occlusion~\cite{Suchan2019OutOS}; and combining neural induction with symbolic deduction produces safer and more traceable planning decisions~\cite{Wei2026ANF}. These works suggest that neuro-symbolic reasoning is a natural requirement for autonomous driving rather than an optional interpretability layer, in line with how the human mind actually works.

Formal traffic-rule representations provide the foundation for this reasoning. Karimi and Duggirala~\cite{Karimi2020FormalizingTR} demonstrate that rules from the California Driver Handbook can be represented in FOL and ASP and executed in CARLA, establishing the practical value of logic-based traffic rules. Irvine \emph{et al.}~\cite{Irvine2023StructuredNL} motivate structured natural language as a precise and stakeholder-accessible interface for defining rules used in verification and validation. Logical English further connects legal text with executable logic: LLM-assisted rule translation is studied in~\cite{DalPont2025YouTT}, context-sensitive legal and violation reasoning is developed in~\cite{Sartor2025ALL}, and a modular Logical-English--Prolog multi-agent framework is presented in~\cite{Sartor2025MindTG}. Collectively, these works demonstrate the importance of formal traffic rules for machine interpretability, human auditability, regulatory compliance, high-level human-like reasoning and the standardization of vehicle behavior.

LLMs provide a flexible interface between natural-language descriptions and autonomous-driving decisions while retaining the strong representational and function-approximation capabilities of deep neural networks. The framework in~\cite{Carvalho2025LLMPoweredFF} highlights their ability to interpret traffic information and produce human-understandable, context-dependent guidance from natural-language scene descriptions. These capabilities make LLMs naturally suited to the neural component of neuro-symbolic reasoning for autonomous driving. The framework in ~\cite{Cui2025CoopReflectTN} use LLMs to convert natural language scene descriptions into vehicle actions. 

Building on these insights, we develop a goal-oriented semantic communication system for LLM-guided neuro-symbolic reasoning in collaborative autonomous driving. Each CAV converts its partial scene description into grounded FOL evidence and selects a task-relevant subset for uplink transmission. The RSU fuses the received evidence, evaluates the traffic-rule set, and derives collaborative safety and right-of-way conclusions. A downlink semantic encoder then selects the deductions to send to corresponding vehicle. The receiving vehicle combines these deductions with its local description to select a high-level action, which is subsequently translated into low-level controls. Thus, the proposed architecture integrates neural processing, formal reasoning, and bandwidth-constrained communication within a single interpretable pipeline.

The cornerstone of the semantic communication system developed in this paper is logical probability. Inspired by the works of Carnap~\cite{c1, c2, carnap_logical_1962} and Hintikka~\cite{c10, Hintikka_combined, c9, Niiniluoto2011TheDO}, we introduce a random-support Dirichlet--Categorical model over the Q-sentences of $\mathcal{L}$, providing a modern statistical interpretation of their works in ~\cite{carnap_logical_1962, c1, c10, c9}. Its support layer represents which Q-sentence types may occur, whereas its frequency layer represents their occurrence frequencies. The model induces logical probabilities for state-descriptions, generalizations, singular formulas, and partial observations. From these probabilities, we define Carnap-style content information, semantic content entropy, and the change in semantic content entropy produced by evidence.

We then formulate goal-oriented semantic communication as a semantic information bottleneck. The encoder selects evidence that maximally reduces uncertainty over the task goals while restricting the use of bandwidth. On the uplink, this criterion prioritizes the observations needed by the RSU to evaluate safety and traffic rules; on the downlink, it prioritizes the conclusions and advisories needed by each vehicle. The framework therefore provides a formal explanation of both what is transmitted and why it is relevant to the driving task.

We evaluate the proposed method on the UCLA MDrive benchmark in CARLA~\cite{Coscoy2026MDriveBC, Dosovitskiy2017CARLAAO}, including occlusion-driven pre-crash scenarios, an unprotected left turn, an intersection deadlock, a highway on-ramp merge, a roundabout navigation, and two minor roads entering a major roads at unsignallized junctions. Under identical communication budgets, semantic selection completes all ten matched scenarios without a collision, whereas uniform information selection causes at least one collision in every scenario and 16 collisions in total. The results show that safe collaborative reasoning depends not only on the availability of formal rules, but also on the reliable communication of the logical evidence required to apply them.

\section{Semantic Communication for Collaborative Perception in Autonomous Driving}
\label{sec:collaborative-perception}
\begin{figure*}[!t]
    \centering
    \includegraphics[width=\textwidth]{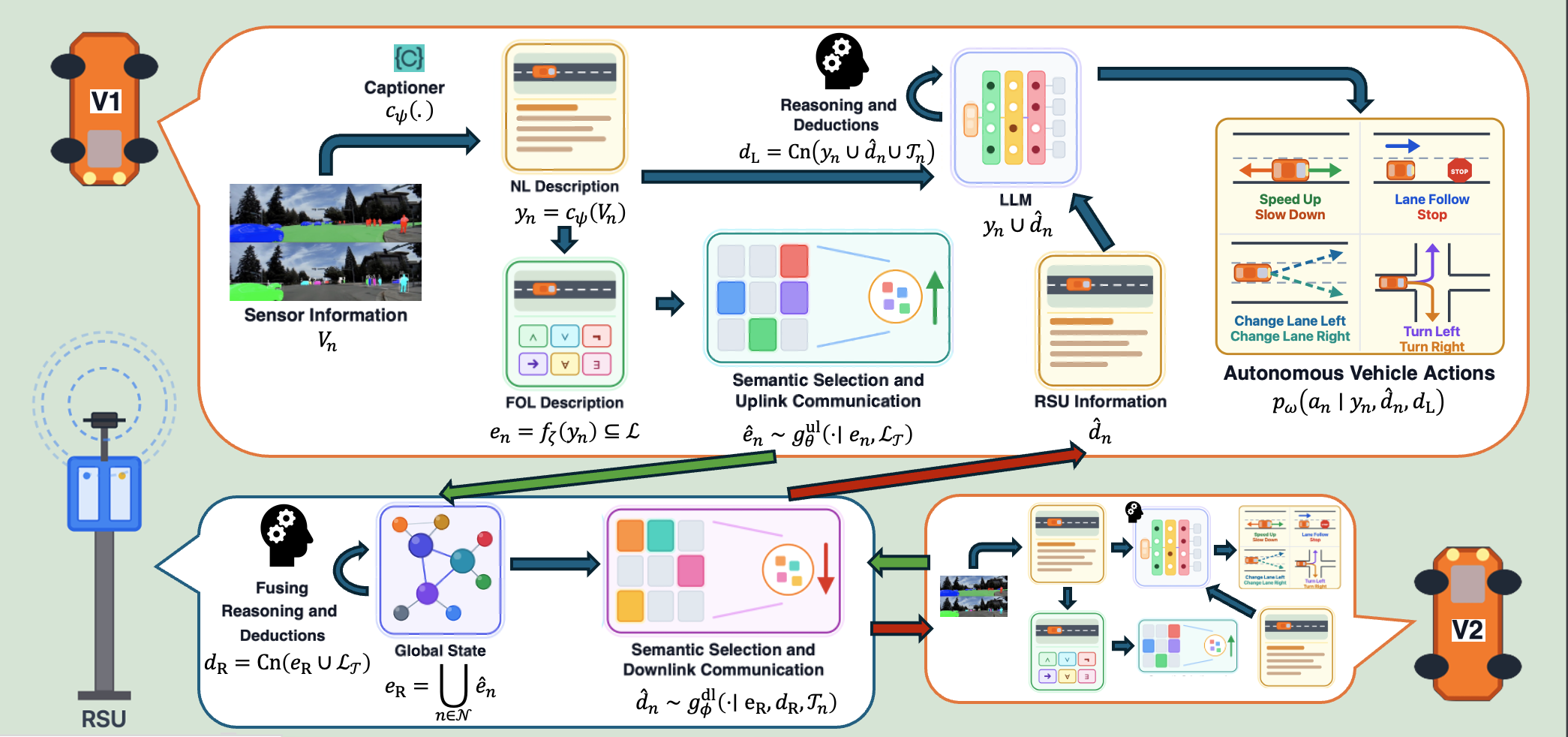}
    \caption{Goal-oriented semantic communication framework for collaborative autonomous driving.}
    \label{fig:wide}
\end{figure*}
We consider $N$ connected autonomous vehicles (CAVs),
$\mathcal{N}=\{1,\ldots,N\}$, and a roadside unit (RSU), as shown in
Fig.~\ref{fig:wide}. Vehicle $n$ observes a partial local perceptual state
$V_n$ using onboard sensors such as cameras, LiDAR, radar, and GPS.
Occlusions, limited sensing range, and viewpoint differences prevent any
single vehicle from observing the complete traffic scene.

To separate semantic communication from low-level perception errors, a captioning interface maps the local perceptual state $V_n$ to a natural-language scene description $ y_n=c_{\psi}(V_n), $ where $c_{\psi}$ denotes the captioner and $y_n$ is its output. Following the abstraction adopted in~\cite{Cui2025CoopReflectTN}, the caption describes observable properties of the ego vehicle, nearby pedestrians, lane geometry, traffic conditions, other vehicles and their states, and other relevant scene attributes, while preserving the partial observability imposed by the vehicle's line of sight. This allows the communication and reasoning components to be evaluated independently of a particular object detector or vision-language perception model. It also provides a common textual representation that can be translated into First-Order Logic (FOL) for semantic communication  and supplied directly to a large language model (LLM) for high-level decision-making. A logic encoder then converts the observation $y_n$ into a corresponding set of local FOL evidences, $ e_n=f_{\zeta}(y_n)\subseteq\mathcal{L}$, where $\Language$ is the FOL language.
An uplink semantic encoder selects task-relevant formulas under the
communication budget:
\begin{equation}
\label{eq:uplink-selection}
\hat e_n
\sim
g_{\theta}^{\mathrm{ul}}
\bigl(\cdot\mid e_n,\mathcal{L}_{\mathcal T}\bigr)
\quad
\text{s.t.}
\quad
\hat e_n\subseteq e_n,
\qquad
\ell(\hat e_n)\leq B_n^{\mathrm{ul}},
\end{equation}
where $\mathcal{L}_{\mathcal T}$ is the set of traffic rules in the form of FOL implication statements,
$\ell(\cdot)$ denotes message length (e.g., number of FOL predicate-entity pairs, numbers of bits, etc.), and $B_n^{\mathrm{ul}}$ is the
uplink budget. The encoder therefore prioritizes formulas relevant to the
driving task rather than transmitting the complete local evidence.

The RSU aggregates the received uplink messages from CAVs into
$
    e_{\mathrm{R}}
    =
    \bigcup_{n\in\mathcal{N}}\hat e_n.
$
The RSU combines this evidence with the task-rule set $\mathcal{L}_{\mathcal{T}}$ and computes the corresponding deductive closure
$
    d_{\mathrm{R}}
    =
    \operatorname{Cn}
    \bigl(e_{\mathrm{R}}\cup\mathcal{L}_{\mathcal{T}}\bigr),
$
where $\operatorname{Cn}(\cdot)$ indicates the set of logical consequences derivable from its argument. That is, the RSU evaluates
the rules whose premises are distributed across vehicles and
derive safety-critical conclusions concerning pedestrians, collisions, hazard warnings, or right-of-way relations. However, those rules with premises supported entirely by a local view may instead be
evaluated locally by CAVs. For each vehicle, the RSU selects an CAV-specific subset of its deductions:
\begin{equation}
\label{eq:downlink-selection}
\hat d_n
\sim
g_{\phi}^{\mathrm{dl}}
\bigl(\cdot\mid e_{\mathrm R}, d_{\mathrm R},\mathcal{T}_n\bigr),
\quad
\text{s.t.}
\quad
\hat d_n\subseteq d_{\mathrm R},
\qquad
\ell(\hat d_n)\leq B_n^{\mathrm{dl}},
\end{equation}
where $\mathcal{T}_n$ is the set of FOL statements related to vehicle $n$'s task and
$B_n^{\mathrm{dl}}$ is its downlink budget. Let
$
\mathcal{J}
$
denote the set of available high-level actions. Based on the local scene
description $y_n$ the vehicle-specific deductions $\hat d_n$ received
from the RSU, and local reasoning $
    d_{\mathrm{L}}
    =
    \operatorname{Cn}
    \bigl(y_n\cup\hat d_n\cup\mathcal{L}_{\mathcal{T}}\bigr),
$ the LLM selects
$
p_{\omega}\!\left(a_n\mid y_n,\hat d_n, d_{\mathrm{L}}\right),
$
where $p_{\omega}$ is a mapping from observations to actions. A PID
controller then translates the selected action $a_n$ into low-level
steering, throttle, and braking.



\section{FOL Representation of the World}
\label{sec:fol-representation}

We represent the traffic environment using a function-free First-Order Logic
(FOL) language $\mathcal{L}$ \cite{nelte}. The vocabulary of
$\mathcal{L}$ consists of a finite set of monadic logic predicates
$\mathcal{P}_1$; a finite set of dyadic logic predicates
$\mathcal{P}_2$; a countably infinite set of variables $\mathcal{X}$; the
logical connectives of negation $\neg$, logical and $\land$, logical or $\lor$, logical implication $\rightarrow$, and logical entailment $\models$; and existential quantifier
$\exists$ and universal quantifier $\forall$. For example,
$\mathit{IsPedestrian}(x_1)$ is a monadic predicate that evaluates to either true or false, depending on whether entity assigned to $x_1$ is truly pedestrian or not, whereas
$\mathit{LeftOf}(x_1,x_2)$ and $\mathit{InIntersection}(x_1,x_2)$ are dyadic
predicates, and $x_1$ and $x_2$ are variables of the predicates.

The traffic rule set $\mathcal{L}_{\mathcal{T}}$ is a finite collection of logical implication formulas
in $\mathcal{L}$. These formulas encode traffic regulations, safety
requirements, and task-specific decision rules. For example,
$
\forall x_1\;(
    \mathit{IsVehicle}(x_1)
    \land
    \exists x_2[
        \mathit{IsPedestrian}(x_2)
        \land
        \mathit{InIntersection}(x_1,x_2)
    ]
    \rightarrow
    \mathit{Stop}(x_1)
)
$
states that a vehicle $x_1$ must stop whenever a pedestrian $x_2$ occupies an
intersection relevant to its path.

In the collaborative-perception framework, the evidence set $e_n$ of
vehicle $n$ contains grounded predicates of $\mathcal{L}$. For example,
$
    e_n
    =
    \{\xi_{n,1},\ldots,\xi_{n,3}\} =
    \{\mathit{IsPedestrian}(p_1),
    \mathit{InIntersection}(v_1,p_1),
    \mathit{LeftOf}(t_1,v_2)\}
$
may encode that entity $p_1$ is a pedestrian, that it lies in the relevant
intersection for vehicle $v_1$, and that truck $t_1$ is in the left of vehicle $v_2$.

\subsection{Q-Sentences}

To construct a finite semantic basis for sentences in $\Language$, we
consider the complete truth assignments to a fixed collection of predicate
occurrences involving two distinct individuals. Let $x_1,x_2\in\mathcal{X}$
and define the set of predicate slots
\begin{equation}
\label{eq:predicate-slots}
\begin{aligned}
\mathcal{A}(x_1,x_2)
={}&
\{P(x_1),P(x_2):P\in\mathcal{P}_1\}
\\
&\cup
\{P(x_1,x_2),P(x_2,x_1):P\in\mathcal{P}_2\}.
\end{aligned}
\end{equation}
Thus, a predicate slot is a particular predicate applied to a particular
ordered tuple of variables. Let
$
    T
    :=
    |\mathcal{A}(x_1,x_2)|
    =
    2|\mathcal{P}_1|+2|\mathcal{P}_2|,
$
and enumerate these slots as
$\mathcal{A}(x_1,x_2)=\{A_1,\ldots,A_T\}$.

\begin{definition}[Q-Sentence\cite{c1, c10}]
\label{def:q-sentence}
For each binary configuration
$\boldsymbol{\delta}_i=(\delta_{i,1},\ldots,\delta_{i,T})
\in\{0,1\}^{T}$, define
\begin{equation}
\label{eq:q-sentence}
    Q_i(x_1,x_2)
    :=
    \bigwedge_{t=1}^{T}
    A_t^{\delta_{i,t}},
\end{equation}
where $ A_t^{\delta_{i,t}}$ is $A_t$, if $\delta_{i,t}=1$, and $\neg A_t$ if $\delta_{i,t}=0$. The formula $Q_i(x_1,x_2)$ is called a \emph{Q-sentence}. It specifies a
complete truth assignment to every predicate slot in
$\mathcal{A}(x_1,x_2)$.
\end{definition}

We denote the set of all Q-sentences by
$
    \QSent
    :=
    \{Q_1,\ldots,Q_K\},
    K:=|\QSent|=2^{T}.
$
The index $i$ may therefore be identified with the binary configuration
$\boldsymbol{\delta}_i$. A Q-sentence describes complete state
of an ordered pair. For example, suppose the language contains the monadic predicate
$\mathit{IsVehicle}$ and the dyadic predicates $\mathit{LeftOf}$ and
$\mathit{Approaching}$, and no other predicates. One possible Q-sentence includes
$
Q_i(x_1,x_2)
={}
\mathit{IsVehicle}(x_1)
\land
\mathit{IsVehicle}(x_2)
\land
\mathit{LeftOf}(x_1,x_2)
\land
\neg\mathit{LeftOf}(x_2,x_1)
\land
\mathit{Approaching}(x_1,x_2)
\land
\mathit{Approaching}(x_2,x_1).
$

\begin{proposition}[Q-Sentence Partition]
\label{prop:q-partition}
For any assignment of two entities (i.e., individuals) $v_1$ and $v_2$ to variables $x_1$ and $x_2$, exactly one Q-sentence in $\QSent$ is satisfied.
Equivalently,
\begin{align}
    \bigvee_{i=1}^{K}Q_i(v_1,v_2)
    &\equiv \top,
    \label{eq:q-exhaustive}\\
    Q_i(v_1,v_2)\land Q_j(v_1,v_2)
    &\equiv \bot,
    \qquad i\neq j.
    \label{eq:q-exclusive}
\end{align}

That particular assignment of entities to variables instantiates that Q-sentence.
\end{proposition}

\begin{proof}
Under a fixed model and valuation, each predicate slot $A_t$ is either true
or false. These truth values determine a unique binary vector in
$\{0,1\}^{T}$ and hence a unique Q-sentence. Any two distinct Q-sentences
differ in at least one predicate slot and therefore cannot both be true.
\end{proof}

The Q-sentences provide a finite semantic basis for formulas over the
predicate slots. Thus,

\begin{proposition}[Disjunctive Normal Form over Q-Sentences]
\label{dnfinfol}
Let $\varphi(x_1,x_2)$ be a formula constructed from the predicate slots in
$\mathcal{A}(x_1,x_2)$. Then
\begin{equation}
\label{eq:q-dnf}
\varphi(x_1,x_2)
\equiv
\bigvee_{\substack{
Q_i\in\QSent\\
Q_i\models\varphi
}}
Q_i(x_1,x_2).
\end{equation}
\end{proposition}

\begin{proof}
By Proposition~\ref{prop:q-partition}, every assignment to
$(x_1,x_2)$ instantiates exactly one Q-sentence. Formula $\varphi$ is therefore
satisfied precisely by the disjunction of the Q-sentences that entail it.
\end{proof}

\subsection{State-Descriptions and Partial Observations}

For a finite collection of observed ordered pairs of entities
$
    \{(v_{t,1},v_{t,2})\}_{t=1}^{M},
$
a complete state-description instantiates exactly one Q-sentence for each pair:

\begin{definition}[State-Description]
\label{def:state-description}
A state-description over $ \{(v_{t,1},v_{t,2})\}_{t=1}^{M}$ is a conjunction
\begin{equation}
\label{eq:state-description}
    \mathrm{SD_j}
    =
    \bigwedge_{t=1}^{M}
    Q_{X_t}(v_{t,1},v_{t,2}),
\end{equation}
where $X_t\in\{1,\ldots,K\}$ denotes the index of the Q-sentence instantiated
by the $t$-th ordered pair. Each different realization
$\boldsymbol{x}^{(j)}=(x_1^{(j)},\ldots,x_M^{(j)})$ of
$(X_1,\ldots,X_M)$ induces one of the $K^M$ possible complete
state-descriptions.
\end{definition}

If Q-sentence $Q_i$ occurs
$m_i$ times in a state-description, then
$
    m_i
    :=
    \sum_{t=1}^{M}\mathbf{1}\{X_t=i\},
    \sum_{i=1}^{K}m_i=M,
$
where $\mathbf{1}\{\cdot\}$ is the indicator function.
In practice, a vehicle rarely observes every predicate slot. A formula
$\xi(x_1,x_2)$ may determine only a subset of the predicate literals in a Q-sentence
and leave the remaining slots unspecified. We associate such a partial
observation with its compatible Q-sentence set
\begin{equation}
\label{eq:observation-cube}
    B(\xi)
    :=
    \left\{
        i\in\{1,\ldots,K\}:
        Q_i(x_1,x_2)\models\xi(x_1,x_2)
    \right\}.
\end{equation}
We refer to $B(\xi)$ as the \emph{observation cube} induced by $\xi$. A
complete observation corresponds to $|B(\xi)|=1$, whereas an incomplete
observation generally satisfies $|B(\xi)|>1$.

For example, observing only
$\mathit{IsPedestrian}(p_1)$ rules out every Q-sentence containing
$\neg\mathit{IsPedestrian}(p_1)$ but leaves all other predicate slots
undetermined. The corresponding cube therefore contains
$2^{T-1}$ Q-sentences. Accordingly, the logical evidence available to vehicle $n$ may be written as
$
    e_n
    =
    \{\xi_{n,1},\ldots,\xi_{n,L_n}\},
$
or, equivalently, as the collection of compatible sets
$
    \mathcal{B}_n
    =
    \{B(\xi_{n,1}),\ldots,B(\xi_{n,L_n})\}.
$

\subsection{Generalizations}

A state-description $\mathrm{SD}_j$ specifies the exact Q-sentence assigned
to every observed pair. We also consider \emph{generalizations} that retain
only the set
$\mathcal{S}_j\subseteq\{1,\ldots,K\}$
of Q-sentence types permitted to occur, without fixing their assignments or
frequencies. We call $\mathcal{S}$ the \emph{support} of the generalization
and denote its width by $w=|\mathcal{S}|$. 

\begin{definition}[Generalization] 
\label{def:generalization} 
For a subset $\mathcal{S}_j\subseteq\{1,\ldots,K\}$, define 
\begin{equation} 
\label{eq:generalization} 
g_j:={} \bigwedge_{i\in\mathcal{S}_j} \exists (x, y)\, Q_i(x,y) \land \bigwedge_{i\notin\mathcal{S}_j} \neg\exists (x, y)\, Q_i(x,y). 
\end{equation} 
Different supports $\mathcal{S}_j$ induce different generalizations.
\end{definition}
Thus, a generalization specifies
which Q-sentence types are possible, while a state-description specifies
how those types are instantiated in finite evidence. This distinction
motivates the support and frequency layers of the random-support
Dirichlet--Categorical model. 
\begin{remark}
    Apart from the generalizations defined above, other universal and existential generalizations can also be formulated in $\Language$. Since the
Q-sentences form a basis for $\Language$ as per Proposition~\ref{dnfinfol} and partition the logical state space
according to Proposition~\ref{prop:q-partition}, each such generalization can
be expressed as a disjunction of the generalizations defined in
\eqref{eq:generalization}.
\end{remark}

\subsection{Inductive Logical Probabilities}

Let $\ProbMeas$ be an inductive logical probability measure defined over $\Language$. Then, the inductive logical probability, a.k.a., the degree of confirmation, is defined as:

\begin{definition}[Degree of Confirmation {\cite{c2}}]
\label{degofconf}
For a FOL sentence $\varphi\in\Language$ and evidence $e\in\Language$ with
$c(e)>0$, the \emph{degree of confirmation} of $\varphi$ by $e$ is
\begin{equation}
\label{eq:degree-confirmation}
c(\varphi\mid e)
:=
\frac{c(\varphi\land e)}{c(e)}.
\end{equation}
\end{definition}

Thus, $c(\varphi\mid e)$ is the conditional inductive logical probability
that $\varphi$ holds given $e$. The evidence-free form is understood as
\begin{equation}
\label{eq:prior-confirmation}
c_0(\varphi)
:=
c(\varphi\mid\top)
=
c(\varphi),
\end{equation}
where $\top$ denotes tautological evidence. 

\begin{remark}

We use $c(\cdot)$ for inductive logical probability, which measures the support one formula receives from another within a specified language and inductive method. Unlike statistical probability, it concerns logical support rather than event frequency, although the model's frequency layer connects the two notions through Q-sentence frequencies.
\end{remark}

Since the
Q-sentences form a basis for $\Language$ as per Proposition~\ref{dnfinfol} and partition the logical state space
according to Proposition~\ref{prop:q-partition}, for a general or singular formula $\varphi$ represented by
compatible Q-sentences, its inductive logical probability is
\begin{equation}
\label{eq:formula-probability}
c(\varphi)
=
\sum_{\substack{
Q_i\in\QSent\\
Q_i\models\varphi
}}
c(Q_i),
\end{equation}
where the sum is valid because the Q-sentences are mutually exclusive. Next, we present a suitable construction of an inductive probability measure $\ProbMeas$ over $\Language$, which we will need to define and quantify semantic information and entropy.

\section{A Random-Support Dirichlet--Categorical Framework for Logical
Probabilities}
\label{sec:rsdc}

In this section, we reinterpret the Carnap--Hintikka inductive systems through a two-layer
generative model over the Q-sentences of $\Language$. Inspired by
Hintikka's constituents, the \emph{support layer} determines which
Q-sentence types may occur; the \emph{frequency layer} models their
frequencies through a Dirichlet--Categorical representation of Carnap's
$\lambda$-continuum. Together, the layers define logical probabilities for
both general and singular formulas while distinguishing a type's existence
from its frequency conditional on existence.

\subsection{The Frequency Layer: Carnap's $\lambda$-System}
\label{subsec:rsdc-frequency}

We first assume full support $\mathcal{S}_K$, so that all $K$ Q-sentences are permitted.
Support uncertainty is introduced in Sec.~\ref{subsec:rsdc-support}.

\begin{definition}[Symmetric Dirichlet--Categorical model]
\label{def:rsdc-dircat}
Let
\[
\boldsymbol{\theta}=(\theta_1,\ldots,\theta_K)
\in
\Delta^{K-1}
:=
\left\{
\boldsymbol{\theta}:
\theta_i\geq0,\ 
\sum_{i=1}^{K}\theta_i=1
\right\}
\]
denote the Q-sentence frequency vector, and let
$X_t\in\{1,\ldots,K\}$ be the Q-sentence index of the $t$-th observed
ordered pair. For $\lambda>0$,
\begin{equation}
\label{eq:rsdc-dircat}
\boldsymbol{\theta}
\sim
\operatorname{Dir}\!\left(
\frac{\lambda}{K},\ldots,\frac{\lambda}{K}
\right),
\;\; X_1,X_2,\ldots\mid\boldsymbol{\theta}
\overset{\mathrm{iid}}{\sim}
\operatorname{Categorical}(\boldsymbol{\theta}).
\end{equation}
\end{definition}

The prior is symmetric as each Q-sentence receives the same parameter
$\lambda/K$, ensuring no Q-sentence type is privileged a priori. We first recall the standard Dirichlet integral
used to obtain the induced probability of state-description by marginalizing
$\boldsymbol{\theta}$.

\begin{lemma}[Dirichlet integral]
\label{lem:rsdc-beta}
For $\beta_i>0$ and
$\bar{\beta}:=\sum_{i=1}^{K}\beta_i$, one can show that
\begin{equation}
\label{eq:rsdc-betaint}
\int_{\Delta^{K-1}}
\prod_{i=1}^{K}\theta_i^{\beta_i-1}
\,d\boldsymbol{\theta}
=
\frac{\prod_{i=1}^{K}\Gamma(\beta_i)}
     {\Gamma(\bar{\beta})}.
\end{equation}
\end{lemma}

For a state-description $\mathrm{SD}_j$ with observation counts
$(m_1,\ldots,m_K)$, conditional independence gives
$
c(\mathrm{SD}_j\mid\boldsymbol{\theta})
=
\prod_{i=1}^{K}\theta_i^{m_i}.
$
Integrating this likelihood under \eqref{eq:rsdc-dircat} and applying
Lemma~\ref{lem:rsdc-beta} with
$\beta_i=m_i+\lambda/K$ yields
\begin{equation}
\label{eq:rsdc-mk}
c(\mathrm{SD}_j)
=
\frac{\Gamma(\lambda)}
     {\Gamma(\lambda+M)}
\prod_{i=1}^{K}
\frac{\Gamma(m_i+\lambda/K)}
     {\Gamma(\lambda/K)}
=
\frac{\prod_{i=1}^{K}(\lambda/K)_{m_i}}
     {(\lambda)_M},
\end{equation}
where
$
(a)_m
:=
a(a+1)\cdots(a+m-1)
=
\frac{\Gamma(a+m)}{\Gamma(a)}
$
is the rising factorial.

\begin{remark}
Equation~\eqref{eq:rsdc-mk} coincides with Carnap's measure
$m_{\lambda}(\cdot)$ over state-descriptions in his $\lambda$-continuum~\cite{c1}. Carnap obtains it from the
sequential rule
\begin{equation}
\label{eq:rsdc-carnaprule}
\Pr(X_t=i\mid X_1,\ldots,X_{t-1})
=
\frac{m_i^{(t-1)}+\lambda/K}
     {(t-1)+\lambda},
\end{equation}
where $m_i^{(t-1)}$ is the number of previous occurrences of $Q_i$.
Multiplication over the sequence produces the denominator
$(\lambda)_M$; the $m_i$ occurrences of type $i$ produce
$(\lambda/K)_{m_i}$. Multiplication over all types therefore recovers
\eqref{eq:rsdc-mk}. Thus, integrating out
$\boldsymbol{\theta}$ is equivalent to multiplying Carnap's predictive
probabilities, and the latent frequency vector $\boldsymbol{\theta}$ provides a modern
statistical representation of his system.
\end{remark}

The Dirichlet representation also makes posterior updating immediate.
Given evidence $e$ with counts $(m_1,\ldots,m_K)$,
\begin{equation}
\label{eq:rsdc-conj}
\boldsymbol{\theta}\mid e
\sim
\operatorname{Dir}\!\left(
m_1+\frac{\lambda}{K},\ldots,
m_K+\frac{\lambda}{K}
\right),
\end{equation}
and hence
\begin{equation}
\label{eq:rsdc-pred}
\Pr(X_{M+1}=i\mid e)
=
\frac{m_i+\lambda/K}{M+\lambda}
=
\frac{M}{M+\lambda}\frac{m_i}{M}
+
\frac{\lambda}{M+\lambda}\frac{1}{K}.
\end{equation}
Predictive probability is therefore a convex combination of the
empirical frequency $m_i/M$ and the symmetric prior $1/K$, with $\lambda$
acting as the total prior pseudo-count. As $M\rightarrow\infty$, the prior
contribution vanishes and the prediction converges to the empirical
frequency.

\subsection{Support Layer and Generalizations}
\label{subsec:rsdc-support}
\label{subsec:rsdc-gen}

The full-support model in Sec.~\ref{subsec:rsdc-frequency} cannot assign
positive probability to a strict universal generalization. For finite
$\lambda>0$, the symmetric Dirichlet distribution satisfies
\begin{equation}
\label{eq:rsdc-interior}
\Pr(\theta_i=0)=0,\qquad \theta_i>0\ \text{a.s.},\quad i=1,\ldots,K.
\end{equation}
Thus, every Q-sentence type occurs almost surely in an infinite sequence,
and any generalization excluding at least one type has probability zero.
For example, $\forall x\,[R(x)\rightarrow B(x)]$ requires the Q-sentence type
containing $R(x)\land\neg B(x)$ never occur (i.e., have zero frequency), which is impossible under
\eqref{eq:rsdc-interior}. We address this limitation by first selecting the
types permitted to occur and then assigning a Dirichlet distribution over
their frequencies.

\begin{definition}[Existence indicators and random support]
\label{def:rsdc-support}
For each Q-sentence index $i\in\{1,\ldots,K\}$, let
\[
    Z_i
    =
    \begin{cases}
        1, & Q_i\text{ is permitted to occur},\\
        0, & Q_i\text{ is ruled out}.
    \end{cases}
\]
The resulting random support is
$
    \mathcal{S}
    :=
    \{i:Z_i=1\}.
$
Conditional on a support of width $w:=|\mathcal{S}|$, the frequency layer is
\begin{equation}
\label{eq:rsdc-hintikka}
(\theta_i)_{i\in\mathcal{S}}
\mid\mathcal{S}
\sim
\operatorname{Dir}\!\left(
    \frac{\lambda}{w},\ldots,\frac{\lambda}{w}
\right),
\qquad
\theta_i=0
\quad\forall i\notin\mathcal{S},
\end{equation}
with
\[
    \Pr(X_t=i\mid\boldsymbol{\theta},\mathcal{S})
    =
    \theta_i.
\]
\end{definition}

Because proper supports $\mathcal{S}\subsetneq\{1,\ldots,K\}$ now receive
positive probability, the model also assigns positive probability to
events of the form $\theta_i=0$ and, therefore, to generalizations that
exclude particular Q-sentence types.

We assign each Q-sentence an existence probability, or \emph{existence
dial}, through
\begin{equation}
\label{eq:rsdc-dials}
Z_i\sim\operatorname{Bern}(\rho_i),
\qquad
\rho_i\sim\operatorname{Beta}(a_i,b_i),
\qquad
\mathbb{E}[\rho_i]
=
\frac{a_i}{a_i+b_i},
\end{equation}
independently across $i$. The resulting model distinguishes two questions:
whether a Q-sentence type can occur, governed by $Z_i$, and how frequently
it occurs conditional on being permitted, governed by $\theta_i$.

\subsection{Probabilities of Generalizations}
\label{subsec:rsdc-gen}

A generalization $g_j$ makes a claim about which Q-sentence types may occur and
hence depends on the support layer $\mathcal{S}_j$.

\begin{definition}[Forbidden set]
\label{def:rsdc-forbidden}
For a generalization $g_j$, define
\begin{equation}
\label{eq:rsdc-forbidden}
\mathcal{F}_{g_j}
:=
\{i:Q_i\text{ is incompatible with }g_j\},
\qquad
f:=|\mathcal{F}_{g_j}|.
\end{equation}
\end{definition}

The generalization is true exactly when its support contains no forbidden
type:
\begin{equation}
\label{eq:rsdc-genmaster}
\models g_j
\iff
\mathcal{S}_j\cap\mathcal{F}_{g_j}=\emptyset,
\qquad
c(g_j)
=
\sum_{\mathcal{S}_j:\,
\mathcal{S}_j\cap\mathcal{F}_{g_j}=\emptyset}
c(\mathcal{S}_j).
\end{equation}
Independence of the existence indicators gives
\begin{equation}
\label{eq:rsdc-genprior}
c(g_j)
=
\prod_{i\in\mathcal{F}_{g_j}}(1-\rho_i)
=
\begin{cases}
(1-\rho)^f, & \rho_i=\rho,\\[3pt]
\dfrac{B(a,b+f)}{B(a,b)},
& \rho\sim\operatorname{Beta}(a,b).
\end{cases}
\end{equation}
If $g_j$ determines the truth values of the predicate slots in the set $\mathcal{Z}_{g_j}$, then
$f=2^{T-|\mathcal{Z}_{g_j}|}$.

For evidence $e$, define the set of counts of observed Q-sentence types as
$
    \mathcal{O}:=\{i:m_i>0\},
    r:=|\mathcal{O}|.
$
The evidence requires $\mathcal{O}\subseteq\mathcal{S}$. Hence, $g_j$ is
falsified if $\mathcal{O}\cap\mathcal{F}_{g_j}\neq\emptyset$; otherwise,
\begin{equation}
\label{eq:rsdc-genpost}
c(g_j\mid e)
=
\mathbf{1}\{\mathcal{O}\cap\mathcal{F}_{g_j}=\emptyset\}
\prod_{i\in\mathcal{F}_{g_j}}
\left(1-\rho_i^{\mathrm{post}}\right),
\end{equation}
where $\rho_i^{\mathrm{post}}$ is derived in
Sec.~\ref{subsec:rsdc-dials} and is evidence conditioned version of existence dials.

\subsection{Singular Sentences and Marginal Probability of Evidence}
\label{subsec:rsdc-marginal}

A singular sentence $e$, such as observations of vehicles, concerns named ordered pairs (i.e., $(v_1, v_2)$) and must be averaged over
both the support and the frequency vector to obtain it's marginal probability:
\begin{equation}
\label{eq:rsdc-triple}
c(e)
=
\sum_{\mathcal{S}}
\Pr(\mathcal{S})
\int_{\Delta_{\mathcal{S}}}
p(\boldsymbol{\theta}\mid\mathcal{S})
p(e\mid\boldsymbol{\theta},\mathcal{S})
\,d\boldsymbol{\theta},
\end{equation}
where $\Delta_{\mathcal{S}}$ is the simplex face selected by
$\mathcal{S}$.

Fix a support $\mathcal{S}$ of width $w=|\mathcal{S}|$. Let the evidence
contain $M$ Q-sentence observations with counts $(m_1,\ldots,m_K)$. Integrating the
categorical likelihood over the conditional Dirichlet distribution yields
\begin{equation}
\label{eq:rsdc-inner}
c(e\mid\mathcal{S})
=
\mathbf{1}\{\mathcal{O}\subseteq\mathcal{S}\}L_w,
\end{equation}
where
\begin{equation}
\label{eq:rsdc-Lw}
L_w
:=
\frac{\Gamma(\lambda)}{\Gamma(\lambda+M)}
\prod_{i\in\mathcal{O}}
\frac{\Gamma(m_i+\lambda/w)}{\Gamma(\lambda/w)}.
\end{equation}
A Q-sentence permitted by the support but unobserved in the evidence has $m_i=0$ and contributes
$\Gamma(\lambda/w)/\Gamma(\lambda/w)=1$ to $L_w$. Since the conditional likelihood $L_w$ depends on
$\mathcal{S}$ only through its width, supports can be grouped by width:
\begin{equation}
\label{eq:rsdc-sizesum}
c(e)
=
\sum_{w=r}^{K}L_w\pi_w,
\qquad
\pi_w
:=
c\bigl(
    |\mathcal{S}|=w,\,
    \mathcal{O}\subseteq\mathcal{S}
\bigr).
\end{equation}

Let $\mathcal{U}:=\{1,\ldots,K\}\setminus\mathcal{O}$ denote the unobserved Q-sentence
types. Any compatible support is
$\mathcal{S}=\mathcal{O}\cup\mathcal{A}$ for some
$\mathcal{A}\subseteq\mathcal{U}$. Under independent inclusion with grouping the supports according to
$y:=|\mathcal{A}|=w-r$ gives,
\begin{equation}
\label{eq:rsdc-poissonbinom}
\pi_w
=
\left(\prod_{i\in\mathcal{O}}\rho_i\right)
\operatorname{PB}_{\mathcal{U}}(w-r),
\end{equation}
where
\[
\operatorname{PB}_{\mathcal{U}}(y)
:=
\sum_{\substack{\mathcal{A}\subseteq\mathcal{U}\\|\mathcal{A}|=y}}
\prod_{i\in\mathcal{A}}\rho_i
\prod_{i\in\mathcal{U}\setminus\mathcal{A}}(1-\rho_i)
\]
is the Poisson--binomial probability mass function of
$J:=\sum_{i\in\mathcal{U}}Z_i$. Therefore,
\begin{equation}
\label{eq:rsdc-expectation}
c(e)
=
\left(\prod_{i\in\mathcal{O}}\rho_i\right)
\sum_{y=0}^{K-r}
\operatorname{PB}_{\mathcal{U}}(y)L_{r+y}
=
\left(\prod_{i\in\mathcal{O}}\rho_i\right)
\mathbb{E}_{J}[L_{r+J}].
\end{equation}

The dependence of $L_{r+J}$ on $\lambda/(r+J)$ creates
\emph{support-size coupling}. Exact evaluation requires the
Poisson--binomial distribution and is infeasible for large $K=2^T$. We
replace the uncertain support width by its posterior mean
$\bar w:=\mathbb{E}[|\mathcal{S}|\mid e]$ to resolve this coupling and define
$
\bar{\lambda}:=\lambda/\bar w.
$
The likelihood then becomes independent of $J$, and
$\sum_{y}\operatorname{PB}_{\mathcal{U}}(y)=1$ gives
\begin{equation}
c(e)
\approx
\left(\prod_{i\in\mathcal{O}}\rho_i\right)
\frac{\Gamma(\lambda)}{\Gamma(\lambda+M)}
\prod_{i\in\mathcal{O}}
\frac{\Gamma(m_i+\bar{\lambda})}
     {\Gamma(\bar{\lambda})}
\label{eq:rsdc-optB}
\end{equation}
with complexity linear in the number of observed Q-sentence types. The sole
approximation is replacing the random width in $L_w$ by $\bar w$.

\subsection{Updating the Existence Dials based on Evidence}
\label{subsec:rsdc-dials}

The posterior existence dial follows from
$
\Pr(Z_i=1\mid e)
=
\frac{
\sum_{\mathcal{S}\ni i}
\Pr(\mathcal{S})c(e\mid\mathcal{S})
}{
\sum_{\mathcal{S}}
\Pr(\mathcal{S})c(e\mid\mathcal{S})
}.
$
If $i\in\mathcal{O}$, i.e., a Q-sentence type observed in the evidence, then every support with nonzero likelihood contains $i$,
and hence
$
i\in\mathcal{O}
\Longrightarrow
\Pr(Z_i=1\mid e)=1.
$

For an unobserved type, two explanations are possible: either the type is
ruled out, with prior probability $1-\rho_i$, or it is permitted but does
not occur in the $M$ observations, with prior probability $\rho_i$. Under
the mean support width $\bar w$, the marginal distribution of an active
coordinate is
$
    \theta_i
    \sim
    \operatorname{Beta}
    \left(
        \bar{\lambda},
        \lambda-\bar{\lambda}
    \right).
$
The probability that an active Q-sentence type remains
unobserved after $M$ observations is
\begin{equation}
\label{eq:rsdc-qN}
q_M
:=
\Pr(m_i=0\mid Z_i=1)
=
\frac{
\Gamma(\lambda)
\Gamma(\lambda-\bar{\lambda}+M)
}{
\Gamma(\lambda-\bar{\lambda})
\Gamma(\lambda+M)
}.
\end{equation}
The quantity $q_M$ decreases monotonically with $M$: as more evidence is
collected, a permitted type becomes progressively less likely to remain
unobserved. Bayes' rule then gives
\begin{equation}
\label{eq:rsdc-D3}
\rho_i^{\mathrm{post}}
:=
\Pr(Z_i=1\mid m_i=0)
=
\frac{
    \rho_i q_M
}{
    \rho_i q_M+(1-\rho_i)
}.
\end{equation}
For small $M$, $q_M\approx1$ and the posterior dial remains close to its
prior value. As $M$ increases, $q_M$ approaches zero and so is the
existence probability of a persistently unobserved Q-sentence type.

The existence dials of Q-sentences can consequently be represented by three groups based on their evidential status:
uniquely observed types, whose posterior dials equal one, i.e., $\rho_i^{\mathrm{post}} = 1$; partially
instantiated types as a result of observing only a subset of predicates in the Q-sentence, represented symbolically by observation cubes introduced in Sec.~\ref{subsec:rsdc-partial}; and
uninstantiated Q-sentence types, i.e., types that are compatible with no
    observation. Under shared prior hyperparameters $(a_0,b_0)$, uninstantiated types share
    the posterior dial 
$
\rho_{\mathrm{unt}}
=
\frac{\rho_0q_M}
     {\rho_0q_M+(1-\rho_0)},
\rho_0=\frac{a_0}{a_0+b_0}.
$
Accordingly,
\begin{equation}
\label{eq:rsdc-barm}
\bar w
=
r
+
\sum_{i\in\mathcal{B}_n}\rho_i^{\mathrm{post}}
+
\left(2^T-r-|\mathcal{B}_n|\right)\rho_{\mathrm{unt}}.
\end{equation}

\begin{remark}[Self-consistent mean support]
\label{rem:rsdc-selfconsistency}
Because $\bar w$ depends on the posterior existence dials, which themselves
depend on the evidence, it is computed iteratively. Starting from
$\bar w^{(0)}$, we evaluate the posterior dials, update
$
    \bar w^{(k+1)}
    =
    \sum_{i=1}^{K}\rho_i^{\mathrm{post},(k)}
$
are updated until convergence.
\end{remark}

These updates yield the three intended behaviors: uniquely observed Q-sentence types
have posterior existence probability 1, partially instantiated Q-sentence types have
intermediate probabilities, and uninstantiated (i.e., never observed) Q-sentence types decrease toward 0 as
evidence accumulates.

\subsection{Soft Counts for Incomplete Evidence}
\label{subsec:rsdc-partial}

As discussed in Section~\ref{sec:fol-representation}, a partial observation is compatible with an observation cube
$B_t\subseteq\{1,\ldots,K\}$ rather than a unique Q-sentence, 
$
    c(X_t\in B_t\mid\boldsymbol{\theta})
    =
    \sum_{i\in B_t}\theta_i,
$
with complete evidence corresponding to $|B_t|=1$. To deal with these types of observations, we introduce soft (fractional) observation counts, and assign these counts based on expectation maximization (EM) updates. Then, the responsibility
and soft count of Q-sentence type $i$ as
\begin{equation}
\label{eq:rsdc-soft}
\gamma_{t,i}
=
\frac{
\mathbf{1}\{i\in B_t\}\theta_i
}{
\sum_{j\in B_t}\theta_j
},
\qquad
\widetilde m_i
=
\sum_{t=1}^{M}\gamma_{t,i},
\qquad
\sum_{i=1}^{K}\widetilde m_i=M.
\end{equation}

The quantity $\widetilde m_i$ is the \emph{soft count} of type $i$. The responsibilities in \eqref{eq:rsdc-soft} depend on
$\boldsymbol{\theta}$ in~\eqref{eq:rsdc-dircat}, which is itself inferred from the soft counts. Under
the symmetric indicator-likelihood approximation, compatible types are
weighted equally:
\begin{equation}
\label{eq:rsdc-equalsplit}
\gamma_{t,i}
=
\frac{\mathbf{1}\{i\in B_t\}}{|B_t|},
\qquad
\widetilde m_i
=
\sum_{t=1}^{M}
\frac{\mathbf{1}\{i\in B_t\}}{|B_t|}.
\end{equation}
Thus, observation $t$ contributes $1/|B_t|$ to every compatible type. The
soft counts are computed in a single EM pass, and the only remaining iteration
is the update of $\bar w$ described in
Remark~\ref{rem:rsdc-selfconsistency}.

Let $\widetilde{\mathcal{O}}:=\{i:\widetilde m_i>0\}$ denote the Q-sentence types implicated by the evidence. Replacing $m_i$ by $\widetilde m_i$ extends the Gamma-function expressions
of Sec.~\ref{subsec:rsdc-marginal} to partial observations. Substituting these into the Gamma-function expressions and using
$\bar{\lambda}=\lambda/\bar w$ gives
\begin{equation}
\label{eq:rsdc-softmarg}
c(e)
\approx
\prod_{i\in\widetilde{\mathcal{O}}}\rho_i
\frac{\Gamma(\lambda)}
     {\Gamma(\lambda+M)}
\prod_{i\in\widetilde{\mathcal{O}}}
\frac{\Gamma(\widetilde m_i+\bar{\lambda})}
     {\Gamma(\bar{\lambda})}.
\end{equation}

Next, we use these probabilities to define semantic content information and entropy.

\section{Inductive Probabilities and Semantic Content Information}
\label{sec:inductive-semantic-information}

The random-support Dirichlet--Categorical framework in
Sec.~\ref{sec:rsdc} induces an inductive logical probability measure
$c(\cdot)$ over $\Language$.

Carnap defines the semantic content information of a statement through its inductive
probability: a statement has greater content when it excludes more
probability mass from the logical state space.

\begin{definition}[Content Information {\cite{c2}}]
\label{cont-info}
The \emph{content information} of a statement $\varphi$ is
\begin{equation}
\label{eq:absolute-content}
\operatorname{cont}(\varphi)
:=
1-c(\varphi).
\end{equation}
Given evidence $e$, its conditional content information is
\begin{equation}
\label{eq:conditional-content}
\operatorname{cont}(\varphi\mid e)
\equiv
\operatorname{cont}(\varphi;e)
:=
1-c(\varphi\mid e).
\end{equation}
\end{definition}

The quantity $\operatorname{cont}(\varphi)$ measures the probability mass
of the logical alternatives excluded if $\varphi$ is true. For example,
consider two propositional atoms $s$ and $r$, producing four equiprobable
states ($s\land r, s\land \neg r, \neg s\land r, \neg s\land \neg r$). The conjunction
$
    \varphi_1\equiv s\land r
$
is true in one state and eliminates the remaining three, giving
$\operatorname{cont}(\varphi_1)=3/4$. In contrast,
$
    \varphi_2\equiv s
$
is true in two states and eliminates two, giving
$\operatorname{cont}(\varphi_2)=1/2$. Thus, less probable statements carry
greater content because their truth excludes more alternatives.

Conditioned on evidence, $\operatorname{cont}(\varphi\mid e)$ measures the
conditional (i.e., residual) content of $\varphi$ after $e$ has been incorporated. If
$c(\varphi\mid e)\approx1$, then $\varphi$ is strongly expected and has low
conditional content. If $c(\varphi\mid e)\approx0$, its truth would be
surprising and would carry high conditional content.

Let
$
    \boldsymbol{\varphi}
    :=
    \{\varphi_1,\ldots,\varphi_L\}
$
be a set of hypotheses (i.e., traffic rules) which are known to both CAVs and RSU in advance. The average \emph{semantic content entropy} of
$\boldsymbol{\varphi}$ is the inductive probability-weighted average content of its
hypotheses:
\begin{equation}
\label{eq:semantic-content-entropy}
\mathrm{H}_{\mathrm{s}}(\boldsymbol{\varphi})
:=
\sum_{i=1}^{L}
c(\varphi_i)\operatorname{cont}(\varphi_i)
\end{equation}

Unlike Shannon entropy, \eqref{eq:semantic-content-entropy} is a linear, or
Gini-type, entropy induced by Carnap's content measure
$\operatorname{cont}(\varphi)=1-c(\varphi)$. Note also that $c(\varphi)$ is the logical probability that a statement $\varphi$ is true. After observing evidence $e$, the remaining semantic content entropy, i.e., semantic conditional content entropy, is
\begin{equation}
\label{eq:conditional-semantic-content-entropy}
\mathrm{H}_{\mathrm{s}}(\boldsymbol{\varphi}\mid e)
:=
\sum_{i=1}^{L}
c(\varphi_i\mid e)
\operatorname{cont}(\varphi_i\mid e)
\end{equation}

The change in content entropy produced by a particular item of evidence
$e$ is
\begin{equation}
\label{eq:semantic-entropy-reduction}
\begin{aligned}
\gamma_{\mathrm{s}}(\boldsymbol{\varphi};e)
&:=
\mathrm{H}_{\mathrm{s}}(\boldsymbol{\varphi})
-
\mathrm{H}_{\mathrm{s}}(\boldsymbol{\varphi}\mid e)\\
&=
\sum_{i=1}^{L}
c(\varphi_i)\operatorname{cont}(\varphi_i)
-
\sum_{i=1}^{L}
c(\varphi_i\mid e)
\operatorname{cont}(\varphi_i\mid e)
\end{aligned}
\end{equation}
A positive value indicates that $e$ concentrates the confirmation
distribution and reduces uncertainty. A negative value is also possible:
the evidence may redistribute probability more uniformly and thereby
increase uncertainty.

\section{Goal-Oriented Semantic Communication}

In our communication setup, each car which transmit evidence to RSU that is most informative in terms of resolving which traffic rule (hypothesis) must be activated. These traffic rules form a hypothesis space and is known by each car and the RSU in advance. Evidence is then \emph{task-informative} precisely when it concentrates posterior probability mass on a small subset of hypotheses or eliminates irrelevant ones. This motivates the following semantic information bottleneck formulation for goal-oriented semantic communication.

\begin{theorem}[Goal-Oriented Semantic Communication Principle]
\label{thm:goal-oriented-semantic}
Let $\boldsymbol{\varphi} = \{\varphi_1, \dots, \varphi_M\}$ be the set of hypotheses (i.e., traffic rules) induced by $L_{\mathcal{T}}$. Given evidence set $e$, the optimal transmitted evidence $\hat e \subseteq e$ for goal-oriented semantic communication is any solution of
\begin{equation}
\label{eq:goal-ib}
\max_{g(\hat e)} \; \mathrm{\gamma_s}(\boldsymbol{\varphi}; \hat{e}) \;-\; \beta\, \mathrm{I}(e; \hat{e}),
\end{equation}
where $\mathrm{I}(\cdot)$ is Shannon mutual information and $g(\cdot)$ is the encoder $g_{\theta}^{\mathrm{ul}}
\bigl(\cdot\mid e_n,\mathcal{L}_{\mathcal T}\bigr)$ or $g_{\phi}^{\mathrm{dl}}
\bigl(\cdot\mid e_{\mathrm R}, d_{\mathrm R},\mathcal{T}_n\bigr)$. 
\end{theorem}

Semantic information bottleneck requires transmitting as little as possible, i.e., second term in~\eqref{eq:goal-ib}, while retaining maximal reduction in uncertainty concerning $\boldsymbol{\varphi}$. Since the hypothesis space $\boldsymbol{\varphi}$ is fixed, the semantic content entropy $\mathrm{H_s}(\boldsymbol{\varphi})$ is constant. Maximizing change in semantic entropy $\mathrm{\gamma_s}(\boldsymbol{\varphi}; \hat{e})$ is therefore equivalent to minimizing the conditional semantic content entropy $\mathrm{H_s}(\boldsymbol{\varphi} \mid \hat{e})$. A single hypothesis contributes $c(\varphi_m, \hat e)\,\mathrm{cont}(\varphi_m \mid \hat e)$ to this entropy, which is small when $\varphi_m$ is nearly certainly true ($c \approx 1$, $\mathrm{cont} \approx 0$) or nearly certainly false ($c \approx 0$, $\mathrm{cont} \approx 1$), and maximized at $c = \mathrm{cont} = 0.5$. The optimal evidence is therefore the one that most sharply resolves the truth or falsity of each hypothesis.

\section{Use of AI Disclosure}

In this work, Claude Code (Sonnet 5) has been utilized for coding and ChatGPT 5.6 (Sol) is used for proofreading. 

\section{Experimental Results}
\label{sec:experimental-results}

\subsection{Experimental Setup}

We evaluate the proposed framework on the MDrive benchmark
\cite{Coscoy2026MDriveBC}, implemented in CARLA 0.9.12
\cite{Dosovitskiy2017CARLAAO}. Our implementation reuses code
from~\cite{Cui2025CoopReflectTN}. Our code and experimental logs are
available on GitHub.\footnote{\url{https://github.com/ahmetfsaz/GOLBSCNSR}}

MDrive evaluates cooperative autonomous driving through perception sharing
and decision negotiation. We focus on the MDrive-Interaction
scenarios, which emphasize inter-vehicle coordination and collision
avoidance under occlusion and limited bandwidth. The selected scenarios
include four two-CAV Pre-Crash configurations (variants \texttt{B},
\texttt{C}, \texttt{E}, and \texttt{F} in MDrive), grounded in the NHTSA pre-crash
scenario typology; a four-CAV unprotected left turn (variant 3 in MDrive); a
three-CAV intersection deadlock (variant 3 in MDrive); a four-CAV highway on-ramp
merge (variant 2 in MDrive); a four-CAV roundabout scenario (variant 2 in MDrive); and two
three-CAV minor-road-to-major-road unsignalized-junction scenarios
(MMUJ, variants 9 and 10 in MDrive).

The traffic-rule corpus contains 152 FOL hypotheses derived from the
California Driver Handbook~\cite{california_dmv_2026_handbook}. We compare
two strategies for selecting FOL atoms under equal uplink and downlink
budgets, $k_{\mathrm{up}}=k_{\mathrm{down}}=k$: \textbf{semantic}, which
ranks evidence according to Theorem~\ref{thm:goal-oriented-semantic}, and
\textbf{uniform}, which selects $k$ candidate atoms uniformly at random.
All CAVs use \texttt{gpt-5.4-nano} with greedy decoding. All trials use
$k=25$, except Pre-Crash C, for which we use $k=40$ because of the larger
amount of relevant information.

\subsection{Evaluation Metrics}

Following the CARLA Leaderboard protocol, the
\emph{Driving Score} is
$
    \mathrm{DS}
    =
    \mathrm{RC}\times\mathrm{IP},
$
where $\mathrm{RC}\in[0,100]$ is route completion and $\mathrm{IP}\in[0,1]$
is the multiplicative infraction penalty. Tracked infractions include collision with a pedestrian ($\mathrm{IP} = 0.5$), vehicle ($\mathrm{IP} = 0.6$), static objects such as curb, guardrail, poles ($\mathrm{IP} = 0.65$), running a red light ($\mathrm{IP} = 0.7$), stop-sign infraction ($\mathrm{IP} = 0.8$), and driving outside the route lane ($\mathrm{IP} = 1 - \frac{\% \text{outside lane} }{100}$). We track
collisions separately as they are most important safety failures.

Communication quality is measured by the \emph{relevant-atom selection
rate}
$
    \mathrm{SR}
    :=
    \frac{R_{\mathrm{sel}}}{R_{\mathrm{avail}}},
$
where $R_{\mathrm{avail}}$ and $R_{\mathrm{sel}}$ are the numbers of
safety-relevant atoms available before and retained after evidence selection for transmission,
respectively. Uplink (UL) relevance is determined from the antecedents of FOL traffic
rules, whereas downlink (DL) relevance is determined from their consequents. We average
$\mathrm{SR}$ over all rounds.

\subsection{Experiment Results}

Table~\ref{tab:results} reports the communication and driving outcomes.
Across the ten semantic trials, the mean uplink and downlink selection
rates are $92.7\%$ and $97.1\%$, respectively, with a mean DS of $93.8$
and no collisions. Across the ten uniform trials, the corresponding values
are $44.5\%$, $56.3\%$, and $46.0$, with 22 collisions. In all ten
same-scenario, same-budget comparisons, semantic selection completes the
trial without a collision, whereas uniform selection produces at least one.

\begin{table}[t]
\centering
\caption{Communication and driving performance.}
\label{tab:results}
\resizebox{\columnwidth}{!}{
\begin{tabular}{@{}llcrrrr@{}}
\toprule
Scenario & Selection & $k$ & UL SR & DL SR & DS & Collision\\
\midrule
Pre-Crash B & Semantic & 25 & 82.5\% & 100.0\% & 100.0 & 0\\
            & Uniform  & 25 & 29.7\% & 0.0\% & 36.1 & 2\\
\addlinespace
Pre-Crash C & Semantic & 40 & 96.0\% & 100.0\% & 100.0 & 0\\
            & Uniform  & 40 & 45.7\% & 100.0\% & 28.2  & 2\\
\addlinespace
Pre-Crash E & Semantic & 25 & 92.4\% & 100.0\% & 84.0  & 0\\
            & Uniform  & 25 & 47.8\% & 78.3\%  & 33.0  & 2\\
\addlinespace
Pre-Crash F & Semantic & 25 & 94.3\% & 100.0\% & 93.5 & 0\\
            & Uniform  & 25 & 49.9\% & 81.8\%  & 73.7 & 1\\
\addlinespace
Left Turn/3 & Semantic & 25 & 92.5\% & 94.4\% & 83.9 & 0\\
            & Uniform  & 25 & 45.6\% & 45.9\% & 32.4 & 2\\
\addlinespace
Deadlock/3  & Semantic & 25 & 93.0\% & 100.0\% & 100.0 & 0\\
            & Uniform  & 25 & 48.6\% & 61.4\%  & 52.7  & 2\\
\addlinespace
On-Ramp/2   & Semantic & 25 & 94.6\% & 94.8\% & 99.0 & 0\\
            & Uniform  & 25 & 43.8\% & 36.9\% & 56.5 & 5\\
\addlinespace
Roundabout/2& Semantic & 25 & 93.8\% & 95.6\% & 85.9 & 0\\
            & Uniform  & 25 & 48.6\% & 52.1\% & 76.6 & 1\\
\addlinespace
MMUJ/9      & Semantic & 25 & 93.1\% & 91.1\% & 91.5 & 0\\
            & Uniform  & 25 & 41.5\% & 53.3\% & 33.9 & 3\\
\addlinespace
MMUJ/10     & Semantic & 25 & 94.6\% & 95.5\% & 100.0 & 0\\
            & Uniform  & 25 & 43.3\% & 52.8\% & 37.2  & 2\\

\bottomrule
\end{tabular}}
\vspace{-4mm}

\end{table}

Semantic uplink selection remains between $82.5\%$ and $96.0\%$ across all
trials, whereas uniform selection never exceeds $49.9\%$. This gap shows
that uniform sampling frequently discards rule-grounding evidence when the
candidate pool exceeds the communication budget. Semantic selection also
retains nearly all actionable RSU conclusions: its lowest downlink
selection rate is $91.1\%$, compared with $36.9\%$ for uniform selection
when relevant downlink atoms are available.

At $k=40$ in Pre-Crash C, semantic selection attains a DS of $100$ with no
collisions, compared with a DS of $28.2$ and two collisions under uniform
selection. In Pre-Crash E, semantic selection attains a DS of $84.0$
without a collision, whereas uniform selection yields a DS of $33.0$ and
two collisions. Among the multi-agent scenarios, the largest driving-score
gap occurs in MMUJ/10, where semantic selection achieves a DS of $100$,
compared with $37.2$ under uniform selection. The highway merge exhibits
the largest collision count under uniform selection, with five collisions,
whereas semantic selection achieves a DS of $99.0$ without a collision.

Whenever semantic selection yields a DS below $100$, the reduction is due
to minor lane infractions, incomplete route completion, or timeout before
route completion. No safety hazards were observed under semantic selection.

\subsection{Failure Analysis}
\label{sec:discussion}
Inspection of the communication and decision traces identifies two
principal failure modes. First, uniform uplink selection may discard
complementary facts required by the RSU to derive a safety warning. In
Pre-Crash C, separation and closing-speed facts failed to co-occur in the
selected uplink during the final approach, preventing timely collision-risk
inference. In Pre-Crash B, vehicle-coordinate atoms were eliminated from
the selection pool despite their relevance to the hypotheses, preventing
the RSU from establishing vehicle relevance and deriving cross-vehicle
conclusions.

Second, a correctly derived warning may be dropped or delivered to only
one member of a conflicting pair. In the uniform deadlock trial, a critical
warning was omitted from the downlink to one vehicle, causing a collision.
In the highway-merge trial, one CAV received the warning while the other
did not; a later warning arrived only after the remaining time to collision
had fallen below the $0.5$\,s decision interval, resulting in a collision.

Trace inspection confirmed that every collision under uniform selection
resulted from the omission of safety-critical information. Relevant
evidence was either excluded from the uplink, preventing the corresponding
safety or right-of-way conclusion from being derived, or omitted during
downlink selection. Across these failures, the driving policy acted
consistently with the information it received; collisions occurred because
safety-critical atoms were unavailable to the decision maker. Semantic
selection cannot recover a conclusion that was never derived, but it
substantially reduces uplink starvation and downlink dropout when relevant
atoms are available under the same communication budget. This mechanism
explains its higher driving scores and zero observed collisions in
Table~\ref{tab:results}.

\section{Conclusion}
We developed an FOL-based semantic communication framework for neuro-symbolic decision-making in collaborative autonomous driving. By combining inductive logical probability, goal-oriented evidence selection, RSU-based deduction, and LLM-guided action selection, the framework communicates information according to its relevance to traffic-rule evaluation. Experiments in CARLA demonstrate that semantic selection reliably preserves safety-critical evidence and substantially improves driving performance over uniform selection under identical communication budgets.

\bibliographystyle{unsrt} 
\bibliography{refs} 

\end{document}